\documentclass[runningheads]{llncs}
\usepackage{todonotes}
\usepackage{caption}
\usepackage{subcaption}
\usepackage{tikz}
\usetikzlibrary{calc}
\usepackage[utf8]{inputenc}      
\usepackage[T1]{fontenc}         
\usepackage[USenglish]{babel}    
\usepackage[intlimits]{amsmath} 
\usepackage{amsfonts}                
\usepackage{amssymb}              
\usepackage[most]{tcolorbox}
\usepackage{complexity}
\usepackage{mathtools}
\usepackage{setspace}
\usepackage{graphicx}
\usepackage{hyperref}
\usepackage{comment}
\usepackage{mathrsfs} 
\usepackage{nicefrac}
\usepackage{cleveref}

\usepackage{algorithm}
\usepackage{algpseudocode}

\newcommand{\TCCEDVS}{{\sc2CCEDVS}}
\newcommand{\TCCVS}{{\sc2CCVS}}

\newcommand{\problemdef}[3]{
\vspace{0.3cm}
\begin{tcolorbox}[enhanced,attach boxed title to top left={yshift=-3.5mm,yshifttext=0mm,xshift=3mm}, colback=gray!3,colframe=black,colbacktitle=white,boxrule=0.5pt,
  title=#1,coltitle=black,fonttitle=\bfseries,
  boxed title style={size=small,colframe=black, boxrule=0.5pt,colback=white} ]
  \vspace{0.3cm}
  \textbf{Given:} \hspace{1.5em} #2 \\[4pt]
  \textbf{Question:} \hspace{0em} #3
\end{tcolorbox}
\vspace{0.3cm}
}

\newtheorem{construction}{Construction}

\begin{document}
\title{On the Complexity of 2-Club Cluster Editing with Vertex Splitting\thanks{This research project was supported by the Lebanese American University under the President’s Intramural Research Fund PIRF0056.}}

\author{Faisal N. Abu-Khzam\inst1
\and Tom Davot\inst2
\and Lucas Isenmann\inst1
\and Sergio Thoumi\inst1}
\authorrunning{Abu-Khzam, Davot, Isenmann and Thoumi}
\institute{
Department of Computer Science and Mathematics\\
Lebanese American University, 
Beirut, Lebanon.\\
\email{\{faisal.abukhzam,lucas.isenmann\}@lau.edu.lb}\\
\email{sergio.thoumi@lau.edu}
\and
LERIA, University of Angers\\
F-49000 Angers, France \\
\email{tom.davot@univ-angers.fr
}
}

\maketitle 

\begin{abstract}

Editing a graph to obtain a disjoint union of $s$-clubs is one of the models for correlation clustering, which seeks a partition of the vertex set of a graph so that elements of each resulting set are close enough according to some given criterion. For example, in the case of editing into $s$-clubs, the criterion is proximity since any pair of vertices (in an $s$-club) is within a distance of $s$ from each other.
In this work, we consider the vertex splitting operation, which allows a vertex to belong to more than one cluster. This operation was studied as one of the parameters associated with the {\sc Cluster Editing} problem. We study the complexity 
and parameterized complexity 
of the {\sc $2$-Club Cluster Edge Deletion with Vertex Splitting} and {\sc $2$-Club Cluster Vertex Splitting} problems. 
We prove that both problems are $\NP$-complete and $\APX$-hard.
On the positive side, we show that the two problems are solvable in polynomial-time on trees and that they are both fixed-parameter tractable with respect to the number of allowed editing operations.

\end{abstract}

\section{Introduction}

Correlation clustering is viewed as a graph modification problem where the objective is to perform a sequence of editing operations (or modifications) to obtain a disjoint union of clusters. Many variants of this problem have been studied in the literature, each with a different definition either of what a cluster is or of the various types of allowed modifications. In the {\sc Cluster Editing} problem, for example, a cluster was defined to be a clique and the allowed editing operations were edge additions and deletions \cite{Cai96,KM86,gramm2005graph}. Later, some relaxation models such as $s$-clubs and $s$-clans emerged as they were deemed ideal models for clustering biological networks \cite{balasundaram2005novel,pasupuleti2008detection}. Subsequent efforts studied overlapping clusters in a graph theoretical context \cite{baumes2005finding,fellows2011graph,abu2018cluster}. In this work, we deal with overlapping communities by performing \textit{vertex splitting}, which allows vertices to be cloned and placed in more than one cluster. The {\sc Cluster Editing with Vertex Splitting} (CEVS) problem was first introduced in \cite{abu2018cluster}, but relaxation models are yet to be combined with vertex splitting. 

The {\sc Cluster Editing} and {\sc Cluster Deletion} problems were shown to be $\NP$-complete in \cite{KM86,shamir2004cluster}. Several other variants of the problem have also been proved to be $\NP$-complete. This includes {\sc Cluster Vertex Deletion} \cite{lewis1980node}, {\sc 2-club Cluster Editing} \cite{liu2012editing}, {\sc 2-club Cluster Vertex Deletion} \cite{liu2012editing}, {\sc 2-club Cluster Edge Deletion} \cite{liu2012editing}, {\sc Cluster Vertex Splitting} \cite{firbas2024complexity}, and {\sc Cluster Editing with Vertex Splitting} \cite{abukhzam2023cluster}.
From a parameterized complexity standpoint, {\sc Cluster Editing}, {\sc Cluster (Edge) Deletion}, and {\sc Cluster Vertex Deletion} are known to be fixed-parameter tractable ($\FPT$) \cite{gramm2005graph,huffner2010fixed}.
The same holds for the two club-variants: {\sc 2-club Cluster Edge Deletion} and {\sc 2-club Cluster Vertex Deletion} \cite{liu2012editing}, while {\sc 2-club Cluster Editing} was shown to be $\W[2]$-hard \cite{figiel20212}. 
Furthermore, the {\sc Cluster Editing with Vertex Splitting} problem has also been shown to be $\FPT$ \cite{abukhzam2023cluster}.

From an approximation standpoint, the {\sc Cluster Editing} and {\sc Cluster (Edge) Deletion} problems are $\APX$-hard and have  $O(\log n)$ approximation algorithms \cite{charikar2005clustering}.
On the other hand, {\sc Cluster Vertex Deletion} has a factor-two approximation algorithm \cite{aprile2023tight}.
To the best of our knowledge, problem variants with $s$-clubs or vertex splitting do not have any known approximation results.

The problems mentioned above are all considered different models of correlation clustering. The $s$-club models were shown to be effective in some networks where a clique could not capture all information needed to form a better cluster \cite{balasundaram2005novel,pasupuleti2008detection}. Vertex splitting proved useful, and in fact essential, when the input data has overlapping clusters, such as in protein networks \cite{nepusz2012detecting}. 
So far, vertex splitting has been only used along with cluster editing. In this paper, we introduce the operation to the club-clustering variant by introducing two new problems: {\sc 2-Club Cluster Vertex Splitting} ({\sc 2CCVS}) and {\sc 2-Club Cluster Edge Deletion with Vertex Splitting} ({\sc 2CCEDVS}). These problems seek to modify a graph into a disjoint union of 2-clubs by performing a series of either vertex splitting only 
({\sc 2CCVS}) or vertex splitting and edge deletion operations ({\sc 2CCEDVS}).

\textbf{Our contribution.} 
We prove that {\sc 2CCVS} and {\sc 2CCEDVS} are $\NP$-complete, even when restricted to planar graphs of maximum degree four and three, respectively. Both problems are also shown to be $\FPT$ and solvable in polynomial-time on trees. We also show that, unless $\P = \NP$, the two problems cannot be approximated in polynomial time with a ratio better than a certain constant greater than one.

\section{Preliminaries}

We work with simple, undirected, unweighted graphs, and adopt common graph theoretic terminology. Let $G=(V,E)$ be a graph where $V$ is the set of vertices and $E$ is the set of edges.
A \emph{path} $P_n$ is a sequence of $n+1$ distinct vertices such that $v_i$ is adjacent to $v_{i+1}$ for each $i \in \{1, 2, \ldots, n\}$. The length of a path in a simple unweighted graph is equal to its number of edges. A \emph{geodesic}, or shortest path, between two vertices in a graph is a path that has the smallest number of edges among all paths connecting these vertices.  A \emph{cycle} is a sequence of three or more vertices $v_1, ..., v_n$ forming a path such that $v_n$ is adjacent to $v_1$.

The \emph{distance} between two vertices $u$ and $v$ of $G$, denoted by $d(u,v)$, is the length of a shortest path between them. 
An \emph{$s$-club} is a (sub)graph such that any two vertices are within distance $s$ from each other. Equivalently, the longest path allowed in an $s$-club is a $P_s$. A \emph{clique}, or a 1-club, in $G$ is a subgraph whose vertices are pairwise adjacent. An $s$-club graph is a disjoint union of connected components, such that each component is an $s$-club. 
The \emph{open neighborhood} $N(v)$ of a vertex $v$ is the set of vertices adjacent to it. The \emph{degree} of $v$ is the number of edges incident to $v$, which is $|N(v)|$ since we are considering simple graphs only. A \emph{cut vertex} is a vertex whose deletion disconnects the graph.

A {\em vertex split} is an operation that replaces a vertex $v$ by two copies $v_1$ and $v_2$ such that $N(v) = N(v_1)\cup N(v_2)$. An \emph{exclusive vertex split} requires that $N(v_1) \cap N(v_2) = \emptyset$. In this paper, we do not assume a split is exclusive, but our proofs apply to this restricted version, which is more important in application domains \cite{abu2021}. 
Therefore, all results that we present also hold for exclusive vertex splitting. In this paper, we introduce the following two problems:

\problemdef{$s$-club Cluster Vertex Splitting}{A graph $G=(V,E)$, along with positive integers $s$ and $k$;}{Can we transform $G$ into a disjoint union of $s$-clubs by performing at most $k$ vertex splitting operations?}

\problemdef{$s$-club Cluster Edge Deletion with Vertex Splitting}{A graph $G=(V,E)$, along with positive integers $s$ and $k$;}{Can we transform $G$ into a disjoint union of $s$-clubs by performing a total of at most $k$ edge deletions and/or vertex splitting operations?}

We work with the problem variant where $s$ is fixed to 2.
For a graph $G$, we define $2ccvs(G)$ (resp., $2ccedvs(G)$) as the minimum length of a sequence of splits (resp., edge deletions and splits) to turn $G$ into disjoint union of 2-clubs. 
For example, for a path $P_5 = v_0, \ldots, v_5$ of length 5, we have $2ccvs(P_5) = 2$ by splitting $v_2$ and $v_3$ and, we have $2ccedvs(P_5) = 1$ by deleting the edge $v_2-v_3$.

\section{The Complexity of {\sc 2-Club Cluster Vertex Splitting} and {\sc 2-Club Cluster Edge Deletion with Vertex Splitting}} 
\label{sec:2CCVS}

\subsection{The Complexity of \textsc{2CCVS}}

Since it is obvious that {\sc 2CCVS} is in $\NP$, we show that it is $\NP$-hard by a reduction from 3-SAT \cite{tovey1984simplified}.

\begin{construction} \label{construction:2ccvs}
    Consider a 3-CNF formula $\phi$, and denote by $M$ the number of clauses and by $\mathcal{V}$ the set of variables. For every variable $V$, we denote by $a(V)$ the number of clauses where $V$ appears and by $C(V)_1, \ldots, C(V)_{a(V)}$ the clauses where $V$ appears. Our construction proceeds as follows (see Figure \ref{fig:construction}):
    
    \begin{enumerate}
    \item [-] For each variable $V$, we create a cycle of length $6a(V)$: $v_1, \ldots, v_{6a(V)}$. 
    \item [-] For every clause $C$ where a variable $V$ appears, let $j$ be the index of $C$ in the (above defined) list of the clauses where $V$ appears.
    We define the vertex $v_c$ as $v_{6(j-1)+1}$ (resp., $v_{6(j-1)+2}$) if its corresponding variable $V$ appears positively (resp., negatively). We denote by $v_{c-1}$ (resp., $v_{c+1}$)  the preceding (resp., following) vertex in the variable cycle. 
    Note that we consider the indices modulo $6a(V)$.
    \item [-] For each clause $C=(U \lor V \lor W)$, connect the vertices $u_{C},v_{C},w_{C}$ into a clique.
    \end{enumerate}
\end{construction}

An example is illustrated in Figure \ref{fig:construction}. 
Observe that the obtained graph is of maximum degree $4$.

\begin{figure}[!h]
    \centering

\begin{tikzpicture}[scale=1.4]
  \begin{scope}[xshift=-3cm]
    \foreach \i in {1,...,12} {
      \node[draw, circle] (x\i) at ({60+30*\i}:1) {};
      \node[font=\tiny] at ({60+30*\i}:1.2) {\i};
    }
    \foreach \i in {1,...,12} {
      \pgfmathtruncatemacro{\nextx}{mod(\i, 12) + 1}
      \draw (x\i) -- (x\nextx);
    }
    \node at (0,0) {$X$};
  \end{scope}
  
  \begin{scope}
    \foreach \i in {1,...,12} {
       \node[draw, circle] (y\i) at ({60+30*\i}:1) {};
      \node[font=\tiny] at ({60+30*\i}:1.2) {\i};
    }
    \foreach \i in {1,...,12} {
      \pgfmathtruncatemacro{\nexty}{mod(\i, 12) + 1}
      \draw (y\i) -- (y\nexty);
    }
    \node at (0,0) {$Y$};
  \end{scope}
  
  \begin{scope}[xshift=3cm]
    \foreach \i in {1,...,12} {
      \node[draw, circle] (z\i) at ({60+30*\i}:1) {};
      \node[font=\tiny] at ({60+30*\i}:1.2) {\i};
    }
    \foreach \i in {1,...,12} {
      \pgfmathtruncatemacro{\nextz}{mod(\i, 12) + 1}
      \draw (z\i) -- (z\nextz);
    }
    \node at (0,0) {$Z$};
  \end{scope}
  
  \draw[red, thick, dashed] (x1) to[bend left=25] (y1);
  \draw[red, thick, dashed] (y1) to[bend left=25] (z1);
  \draw[red, thick, dashed] (x1) to[bend left=30] (z1);
  
  \draw[blue, thick, dashed] (x7) to[bend right=25] (y7);
  \draw[blue, thick, dashed] (y7) to[bend right=25] (z8);
  \draw[blue, thick, dashed] (x7) to[bend right=30] (z8);
\end{tikzpicture}

    \caption{The graph $G$ for  $\phi=(X\lor Y \lor Z)\land (X\lor Y \lor \overline{Z})$. The edges of the clause $(X \lor Y \lor Z)$ are in red and the edges of the clause $(X \lor Y \lor \overline{Z})$ are in blue.}
    \label{fig:construction}
\end{figure}

\begin{theorem} \label{th:npc}
{2CCVS} is $\NP$-complete even on graphs with maximum degree four.
\end{theorem}

\begin{proof}
Let $G$ be the graph obtained by Construction~\ref{construction:2ccvs}.
We set $k=11M$ where $M$ is the number of clauses.

$\textit{Claim}.$ The graph $G$ has a sequence of at most $k$ splits 
that turn $G$ into a disjoint union of $2$-clubs if and only if $\phi$ is satisfiable.

($\Leftarrow$) 
Suppose that $\phi$ is satisfiable and consider a satisfying assignment $\sigma$.
Let $V$ be a variable.
If $V$ is true (resp., false), then we split $v_{2j}$ (resp., $v_{2j+1})$ for every $j \in [1,3a(V)]$ so that it separates the cycles in a disjoint union of paths of length 2.
We recall that the indices of the vertices $v_j$ are considered modulo $6a(V)$ in $[1,6a(V)]$.
For each split vertex, one copy will get all the edges coming from outside of the cycle and the other none. 
Therefore, we have already used $\sum_{i\in |V|} 3a(V_i) = 9M$ splits.

Let $C$ be a clause containing three variables $X$, $Y$ and $Z$.
As $\phi$ is satisfied by $\sigma$, then there exists a variable which satisfies $C$.
Without loss of generality, we can suppose that $X$ satisfies $C$.
If $X$ is true, then the vertices $x_{C-1}$ and $x_{C+1}$ are split.
We split the vertices $y_C$ (resp., $z_C$) so that one copy is adjacent to vertices of the clause and the other copy is adjacent to vertices of the cycle of $Y$ (resp., $Z$).
For each clause, we make $2$ splits.
In total, we use $k=2M + 9M$ splits.

The connected components of the resulting graphs are paths of length 2 (in the variable cycles) and stars centered on $x_C$ vertices where $X$ is the chosen variable satisfying a clause $C$.
We conclude that this sequence of $k$ splits leads to a disjoint union of $2$-clubs.
\newline

($\Rightarrow$)
Conversely, suppose we have a sequence of $k$ splits turning $G$ into a disjoint union of $2$-clubs.
Let $V_i$ be a variable and
suppose that less than $3a(V_i)$ vertices of the cycle are split.
As there are $6a(V_i)$ vertices, then there exists two adjacent vertices $v_i$ and $v_{i+1}$, with $i \in [1, 6a(V_i)]$ which are not split.
Thus, the vertices $v_{i-1}$ and $v_{i+2}$ will be at distance at least $3$ in the resulting graph, a contradiction.
Hence, at least $3a(V_i)$ of the vertices of the cycle must be split.

Let us prove that each clause needs $2$ splits.
Consider a clause $C_j$ with variables $X,Y,$ and $Z$.
In the resulting graph, there is a copy $x_c'$ of $x_c$ and a copy $y_c'$ of $y_c$ such that $x_c'$ and $y_c'$ are adjacent.
Suppose that $x_c'$ is adjacent to $x_{c-1}$ or $x_{c+1}$ and $y_c'$ is adjacent to $y_{c-1}$ or $y_{c+1}$.
This would imply that the resulting graph contains two vertices at distance $3$ from each other, a contradiction.
Therefore, either $x_c'$ is not adjacent to $x_{c-1}$ and $x_{c+1}$ or $y_c'$ is not adjacent to $y_{c-1}$ and $y_{c+1}$.
Without loss of generality, we can suppose that $x_c'$ is not adjacent to $x_{c-1}$ and $x_{c+1}$.
By considering the edge $y-z$, we prove in the same way that either there is a copy $y_c'$ of $y_c$ which is not adjacent to $y_{c-1}$ and $y_{c+1}$ or there is a copy $z_c'$ of $z_c$ which is not adjacent to $z_{c-1}$ and $z_{c+1}$.
So we need at least 2 splits for every clause which cannot be used to solve the conflicts in the cycle.

Therefore, we need at least $ 2M +9M$ splits.
As $k = 11M$, we deduce that for each variable $V$ exactly $3a(V)$ vertices of its cycle are split.
Let us prove that in every variable cycle, the vertices which are split are not adjacent.
Let $V$ be a variable, suppose that two adjacent vertices are split.
Then, without loss of generality, we can suppose that $v_1$ and $v_2$ are split.
We still have to deal with the geodesics $v_1, v_{a(V)},v_{a(V)-1}, v_{a(V)-2}$ and $v_2, v_3,v_4,v_5$.
Thus, we need $3a(V)-1$ more splits.
It contradicts the fact that we can only use $3a(V)$ splits to resolve all the $P_3$s of the cycle since we used two splits previously.
Therefore, the split vertices disconnecting the cycle of $v$ are not adjacent.
As half of the vertices of the cycles must be split, we deduce that there are only 2 ways to split the cycle.

Let us define a truth assignment of the variables.
If $v_2$ is split, then we set $V$ to true; otherwise, we set it to false.

Let us prove that the 3-SAT formula is satisfied by the previous assignment.
Let $C$ be a clause with variables $X,Y,Z$.
Without loss of generality we can suppose that $X,Y,$ and $Z$ appear positively in the clause.
Suppose that $X,Y,$ and $Z$ are set to false.
Therefore, $x_1$ is split and the edges $x_1-x_{6a(X)}$ and $x_1-x_2$ are split.
Therefore, the following paths are in the obtained graph: $x_c-x_{c+1}-x_{c+2}$ and $x_c-x_{c-1}-x_{c-2}$.
As $y_c-x_c-x_{c+1}-x_{c+2}$ is a path and $z_c-x_c-x_{c+1}- x_{c+2}$ is a path as well, we deduce that we have to split $x_1$ one more time.
In the same way, we prove that $y_c$ and $z_c$ must be split at least $2$ times. 
This contradicts the fact that the sequence of splits is of length at most $k = 2M+9M$.
We conclude that $C$ is satisfied and that we have a yes-instance. This completes the proof.
\end{proof}

\subsection{The Complexity of \textsc{2CCEDVS}}

\begin{construction} \label{construction:2ccedvs}
    Consider a 3-CNF formula $\phi$, and denote by $M$ the number of clauses and by $\mathcal{V}$ the set of variables. For every variable $V$, we denote by $a(V)$ the number of clauses where $V$ appears and by $C(V)_1, \ldots, C(V)_{a(V)}$ the clauses where $V$ appears. Our construction proceeds as follows (see Figure \ref{fig:construction_2ccedvs}):
    
    \begin{enumerate}
    \item [-] For each variable $V$, we create a cycle of length $6a(V)$: $v_1, \ldots, v_{6a(V)}$. 
    \item [-] For every clause $C$, we create a vertex $c$.
    For each variable $V$ appearing in $C$, let $j$ be the index of $C$ in the (above defined) list of the clauses where $V$ appears.
    We define the vertex $v_c$ as $v_{6(j-1)+1}$ (resp., $v_{6(j-1)+2}$) if its corresponding variable $V$ appears positively (resp., negatively). We denote by $v_{c-1}$ (resp., $v_{c+1}$)  the preceding (resp., following) vertex in the variable cycle. 
    Note that we consider the indices modulo $6a(V)$.
    \item [-] For each clause $C=(U \lor V \lor W)$, connect the vertices $u_{C},v_{C},w_{C}$ to $c$.
    \end{enumerate}
\end{construction}

\begin{figure}[!h]
    \centering

\begin{tikzpicture}[scale=1.4]
  \begin{scope}[xshift=-3cm]
    \foreach \i in {1,...,12} {
      \node[draw, circle] (x\i) at ({60+30*\i}:1) {};
      \node[font=\tiny] at ({60+30*\i}:1.2) {\i};
    }
    \foreach \i in {1,...,12} {
      \pgfmathtruncatemacro{\nextx}{mod(\i, 12) + 1}
      \draw (x\i) -- (x\nextx);
    }
    \node at (0,0) {$X$};
  \end{scope}
  
  \begin{scope}
    \foreach \i in {1,...,12} {
       \node[draw, circle] (y\i) at ({60+30*\i}:1) {};
      \node[font=\tiny] at ({60+30*\i}:1.2) {\i};
    }
    \foreach \i in {1,...,12} {
      \pgfmathtruncatemacro{\nexty}{mod(\i, 12) + 1}
      \draw (y\i) -- (y\nexty);
    }
    \node at (0,0) {$Y$};
  \end{scope}
  
  \begin{scope}[xshift=3cm]
    \foreach \i in {1,...,12} {
      \node[draw, circle] (z\i) at ({60+30*\i}:1) {};
      \node[font=\tiny] at ({60+30*\i}:1.2) {\i};
    }
    \foreach \i in {1,...,12} {
      \pgfmathtruncatemacro{\nextz}{mod(\i, 12) + 1}
      \draw (z\i) -- (z\nextz);
    }
    \node at (0,0) {$Z$};
  \end{scope}

  \node[draw, circle] (C1) at (0,2) {$C_1$};
   \node[draw, circle] (C2) at (0,-2) {$C_2$};
  
  \draw[red, thick, dashed] (x1) to (C1);
  \draw[red, thick, dashed] (y1) to (C1);
  \draw[red, thick, dashed] (z1) to (C1);
  
  \draw[blue, thick, dashed] (x7) to[bend right = 25] (C2);
  \draw[blue, thick, dashed] (y7) to (C2);
  \draw[blue, thick, dashed] (z8) to[bend left=25] (C2);
\end{tikzpicture}
    
    \caption{Example of a graph obtained by Construction~\ref{construction:2ccedvs} from the 3-CNF $\phi = (X \lor  Y \lor Z )\land (X \lor Y \lor \overline{Z})$.}
    \label{fig:construction_2ccedvs}
\end{figure}

\textit{Proof}. Since its membership in $\NP$ is obvious, we proceed by proving the problem is $\NP$-hard by a reduction from the $\NP$-hard 3-SAT problem \cite{tovey1984simplified}.

Given a 3-SAT instance consisting of a 3-CNF formula $\phi$ where each clause has exactly three distinct variables. 
We consider the graph $G$ obtained by Construction~\ref{construction:2ccedvs}.

$\textit{Claim}.$ The constructed instance of  \TCCEDVS{} is a yes-instance if and only if the given 3-SAT input formula $\phi$ is satisfiable.

($\Rightarrow$) Suppose we have a yes-instance of 3-SAT. Let $V$ be our variable and $G_i$ be its corresponding gadget.
For each $1\leq j \leq 2a(V)$: if $V$ is set to true, we delete the edges $v_{3j+1}-v_{3j+2}$. If $V$ is set to false, we delete the edges $v_{3j}-v_{3j+1}$. 
This transforms $G_i$ into a disjoint union of $P_2$s by using $2a(V)$ modifications. 
In total, we have $18M$ edges in all variable gadgets combined so we use $6M$ modifications to transform all variable gadgets into a disjoint union of $P_2$s.

Obviously, the clause gadget by itself is a 2-club, however, the connections with the variable gadgets pose a conflict. To solve this conflict, we do the following: for each clause, there is at least one true literal (if there is multiple, pick one arbitrarily), we set the edge connecting the clause gadget to the corresponding variable gadget as permanent and delete the other two.
In total, we use $2M$ modifications for the rest of the graph.
So, we use $6M+2M=8M$ modifications for the graph.

($\Leftarrow$) Conversely, suppose we have a yes-instance of 2CCEDVS.

We first show that to transform the given instance into a 2-club graph using at most $8M$ modifications, we would need to apply the same modifications that were performed above.
If we were to solve this conflict via splitting, then the first split would turn the cycle of length $l$ to a path of length $l+1$, but deleting an edge leaves us with a path of length $l$. In general, splitting a vertex $v_i$ would always result in a path that is an edge longer than when we delete the edge that is incident to $v_i$.
Thus, the most optimal solution would be to 
delete every third edge in the cycle. 
Next, we have the conflict between the variable gadgets and clause gadgets. 
We have two options: deleting the two of the edges connecting the clause gadget to the variable gadgets, or splitting the clause vertex twice. Splitting the vertex does not lead to the optimal solution unless all literals in the clause are true (otherwise, we would get a 3-club). Thus, the optimal solution would delete the two edges. 
Since we have a yes-instance and at least $k=8M$ modifications are needed to transform the constructed instance into a disjoint union of 2-clubs, then there is a way to modify the graph such that for each clause, one vertex is part of the structure shown in Figure \ref{struct} . If these vertices appear positively, then we set them to true; otherwise, we set them to false. The remaining variables are arbitrarily set to true or false. This concludes the proof.

\begin{figure}[!htb]
    \centering
\begin{tikzpicture}[main/.style = {draw, circle, node distance=2cm, scale = 0.80, minimum size=2em},
every edge quotes/.style = {auto=left, sloped, font=\scriptsize, inner sep=1pt}
] 
\node[main] (1) {};
\node[main] (2) [right of=1] {};
\node[main] (3) [right of=2] {};
\node[main] (4) [below of=2] {$C_j$};
\draw (1) -- node[above, sloped, pos=0.5] {} (2);
\draw (2) -- node[above, sloped, pos=0.5] {} (3);
\draw (2) -- node[above, sloped, pos=0.5] {} (4);

\end{tikzpicture}
    \caption{Clause gadget (becoming a 2-club) after applying the modifications.}
    \label{struct}
\end{figure}

Our proof also implies the $\NP$-hardness of the {\sc 2-Club Cluster Edge Deletion} problem, simply because our proof does not use vertex splitting (although it is allowed). The problem is already known to be $\NP$-complete \cite{liu2012editing}, but we note that the above yields an alternative proof.

\begin{corollary}
    {\sc 2-Club Cluster Edge Deletion} is $\NP$-complete.
\end{corollary}

\begin{theorem} \label{theorem:2CCEDVS_NPC}
The {\sc 2-Club Cluster Edge Deletion with Vertex Splitting} problem is $\NP$-complete.
\end{theorem}

\subsection{\textsc{2CCED} and \textsc{2CCEDVS} on Planar Graphs}

As the two constructions result in planar graphs if the bipartite graph of the 3-SAT instance is planar, we deduce the following from  the $\NP$-hardness of {\sc Planar-3-SAT} \cite{lichtenstein1982planar}:

\begin{corollary}
\label{cor2}
{\sc 2-club Cluster Vertex Splitting} remains $\NP$-complete on planar graphs with maximum degree four.
\end{corollary}

\begin{corollary}
\label{cor3}
 {\sc 2-club Cluster Edge Deletion with Vertex Splitting} remains $\NP$-complete on planar graphs with maximum degree three.
\end{corollary}

\begin{proof}
Consider an instance of Planar-3-SAT, a variant of 3-SAT where the bipartite graph of the instance is planar.
The graph produced by the previous construction can be modified as follows:
    
\begin{itemize}
        \item Replace every variable vertex by a cycle of a certain length.
        \item Replace every clause vertex by a triangle.
        \item For every clause, merge the three vertices of the triangle with one vertex of each variable cycle appearing in the clause.
    \end{itemize}
    Each of these operations preserves the planarity of the graph.
    We deduce that the modified graph is planar and the previous construction gives a reduction from Planar-3-SAT to {\sc 2CCVS} (resp., \textsc{2CCEDVS}) restricted to planar graphs with maximum degree four (resp., three).
    The proof is now complete, given that Planar-3-SAT is \NP-complete \cite{lichtenstein1982planar}.
\end{proof}

Moreover, knowing that the previous construction is linear in the number of vertices and (resp., Planar) 3-SAT does not admit a $2^{o(n+m)} n^{O(1)}$ (resp., $2^{o(\sqrt{n+m})} n^{O(1)}$) time algorithm \cite{cygan2015lower} unless the Exponential Time Hypothesis (ETH) fails, we conclude the following:

\begin{corollary}
\label{cor4}
Assuming the ETH holds, there is no $2^{o(n+m)} n^{O(1)}$-time (resp., $2^{o(\sqrt{n+m})} n^{O(1)}$)-time algorithm for {\sc 2CCVS} (resp., on planar graphs) with maximum degree four.
\end{corollary}

\begin{corollary}

Assuming the ETH holds, there is no $2^{o(n+m)} n^{O(1)}$-time (resp., $2^{o(\sqrt{n+m})} n^{O(1)}$)-time algorithm for \textsc{2CCEDVS} (resp., on planar graphs) with maximum degree three.
\end{corollary}

\subsection{The Complexity of \TCCVS{} on Trees.}
    
\begin{theorem}  \label{theorem:2CCVS_trees}
\textsc{2CCVS} is solvable in polynomial time on trees.
\end{theorem}

\begin{lemma} \label{lemma:tree}
    Let $T$ be a tree, and let $v$ be a leaf vertex of $T$.
    If $w$ is the unique neighbor of $v$ in $T$ and $k$ is the number of neighbors of $w$ which are not leaves (see Figure \ref{fig:tree-example}), then $2ccvs(T) = 2ccvs(T\setminus \{v,w\}) + k$.
\end{lemma}

\begin{proof}
    Let $S$ be a sequence of $p$ splits of $T \setminus \{v,w\}$.
    We get a disjoint union of $2$-clubs by applying this sequence to $T$ and then by splitting the $k$ neighbors of $w$ which are not leaves such that each split separates $w$ from the graph $T \setminus \{v,w\}$ while resulting in a disjoint 2-club containing $w$.
    The sequence of splits is of length $p+k$.

    Let $S'$ be a sequence of splits of $T$ of length $p$ leading to a disjoint union of $2$-clubs.
    As splitting $v$ is useless (because $v$ is of degree $1$), we can suppose that $v$ is not in the sequence.
    Suppose that $w$ is in the sequence.
    Let us show that we can find another sequence with $p$ splits such that $w$ is not split.
    Let $w_0, \ldots , w_r$ be the copies of $w$ (the vertex is split $r$ times).

    \begin{figure}[!h]
        \centering
        
\begin{tikzpicture}
    \node[circle, draw, fill=white] (v) at (0,0) {$v$};
    \node[circle, draw, fill=white] (w) at (2,0) {$w$};
    \node[circle, draw, fill=white] (x) at (4,0) {$x$};

    \node[circle, draw, fill=white, minimum size=0.5cm] (w1) at (2,-1.5) {};
    \node[circle, draw, fill=white, minimum size=0.5cm] (w2) at (2,1.5) {};
    \node[circle, draw, fill=white, minimum size=0.5cm] (w3) at (1,1.5) {};
    \node[circle, draw, fill=white, minimum size=0.5cm] (w4) at (1,-1.5) {};

    \node[circle, draw, fill=white, minimum size=0.5cm] (x1) at (4,1.5) {};
    \node[circle, draw, fill=white, minimum size=0.5cm] (x2) at (4,-1.5) {};

    \node[circle, draw, fill=white, minimum size=0.5cm] (w1a) at (3,-2.5) {};
    \node[circle, draw, fill=white, minimum size=0.5cm] (w1b) at (1,-2.5) {};

    \node[circle, draw, fill=white, minimum size=0.5cm] (w2a) at (3,2.5) {};
    \node[circle, draw, fill=white, minimum size=0.5cm] (w2b) at (1,2.5) {};

    \draw (v) -- (w) -- (x);
    \draw (w) -- (w1);
    \draw (w) -- (w2);
    \draw (w) -- (w3);
    \draw (w) -- (w4);
    \draw (x) -- (x1);
    \draw (x) -- (x2);

    \draw (w1) -- (w1a);
    \draw (w1) -- (w1b);
    \draw (w2) -- (w2a);
    \draw (w2) -- (w2b);

\end{tikzpicture}
\caption{Example of a tree with a leaf $v$ and its neighbor $w$. The vertex $x$ is a non-leaf neighbor of $w$. The vertex $w$ has $3$ neighbors which are not leaves.}
\label{fig:tree-example}
\end{figure}

Suppose that there exists a copy $w'$ of $w$ such that there exists $x_1$ and $x_2$ two neighbors of $w$ which are not leaves such that $w'$ is connected to a copy $x_1'$ of $x_1$ and to a copy $x_2'$ of $x_2$.
Then, the copy $x_1'$ of $x_1$ is not connected to any of the neighbors of $x_1$ other than $w$, otherwise this neighbor would be at distance $3$ to $x_2'$.
In the same way, the copy $x_2'$ of $x_2$ is not connected to any of the neighbors of $x_2$ other than $w$.

If $w$ is split more than $k$ times, then we can just undo the splitting of $w$ (or {\em unsplit} it) and split the non-leaf neighbors of $w$ one more time by creating a copy of $x$ which is only connected to $w$ and delete the edges from the other copies of $x$ to $w$.
Otherwise, $w$ is split less than $k$ times.
Let $w_0$ be a copy of $w$ which is adjacent to $v$.
Consider another copy $w'$ of $w$.
Suppose that $w'$ is adjacent to more than two copies of non-leaf neighbors.
Consider $x$ and $y$ two of these non-leaf neighbors.
As the resulting graph is disjoint union of $2$-clubs, then these copies are not adjacent to other vertices than $w'$.
Thus, we can delete the copy $w'$ and connect its neighbors to $w_0$.
Therefore, we can suppose that $w$ is split less than $k$ times and that all the copies are connected to at most $1$ copy of a non-leaf neighbor of $w$.
This is a contradiction because $w$ would need to be split at least $k$ times, one time for each non-leaf neighbor.

We can therefore suppose that $v$ and $w$ are not split.
Then, there is a sequence of $p$ splits of $T \setminus \{v,w\}$ leading to a disjoint union of $2$-clubs.
Let $x$ be a neighbor of $w$ which is not a leaf.
Then $x$ has a neighbor $y$ which is different from $w$.
There exists a copy $x'$ of $x$ which is connected to $w$.
If this copy is also connected to a copy $y'$ of $y$, then $y'$ would be at distance $3$ from $v$ in the resulting graph, a contradiction.
We deduce that $x'$ is only connected to $w$.
Thus, we can find a sequence of splits of the branch at $x$ of the tree rooted in $v$ with one less split.
As we can do it for every neighbor of $w$ which is not a leaf, we can find a sequence of $p-k$ splits of $T \setminus \{v,w\}$ leading to a disjoint union of $2$-clubs. We conclude that $2ccvs(T) = 2ccvs(T \setminus \{v,w\}) + k$.

As we can find a leaf in linear time, recursively applying the formula of Lemma~\ref{lemma:tree} on a tree $T$ with $n$ vertices leads to computing $2ccvs(T)$ in $O(n^2)$ time.
\end{proof}

\subsection{The Complexity of \TCCEDVS{} on Trees.}

\begin{lemma} \label{lemma:2ccedvs_rec_simple}

Let $y$ be a cut vertex of a graph $G$ such that $A\cup y$ is a 2-club, where $A$ is a connected component in $G\setminus y$ containing a vertex at distance 2 from $y$. If $y$ is adjacent to exactly one vertex in $B = V(G) \setminus (\{y\} \cup A)$, then $2ccedvs(G) = 2ccedvs(G - y -A) + 1$.
\end{lemma}
\begin{proof}
Let $a$ be a vertex in $A$ at distance 2 from $y$ and $x$ be a vertex of $N(y) \cap N(a)$ and $z$ be the only neighbor of $y$ in $B$ (see Figure~\ref{fig:2ccedvs_tree_lemma1}).
    
    \begin{figure}[h] 
    \centering
    \begin{tikzpicture}
    [
        yscale=-1,
        node_style/.style={circle, draw, minimum size=0.5cm} 
    ]
    \def\fscale{5}  \node[node_style] (0) at ($\fscale*(-0.5, 0)$) {$a$};
	 \node[node_style] (1) at ($\fscale*(-0.25, 0.24)$) {$w$};
	 \node[node_style] (2) at ($\fscale*(-0.25, -0.24)$) {$x$};
	 \node[node_style] (3) at ($\fscale*(-0.02, 0)$) {$y$};
	 \node[node_style] (4) at ($\fscale*(0.27, 0)$) {$z$};
	 \node[node_style] (5) at ($\fscale*(0.5, -0.24)$) {};
	 \node[node_style] (6) at ($\fscale*(0.5, 0.24)$) {};

	 \draw (0) to  (1);
	 \draw (1) to  (2);
	 \draw (2) to  (0);
	 \draw (2) to  (3);
	 \draw (3) to  (1);
	 \draw (3) to  (4);
	 \draw (4) to  (5);
	 \draw (4) to  (6);
\end{tikzpicture}
    \caption{Example of a graph with a cut vertex $y$ connected to only one vertex $z$ in $B$. The optimal solution consists here in deleting the edge $y-z$ in any case. }
    \label{fig:2ccedvs_tree_lemma1}
\end{figure}
    
Suppose that $y$ is adjacent to exactly one vertex $z$ in $B$.
Let $\sigma$ be a sequence of edge deletion(s) and vertex split(s) in  $G-y-A$.
We add an edge deletion of $y-z$ at the beginning of $\sigma$.
Then, this new sequence is of length one more than $\sigma$ and it turns $G$ into a disjoint union of 2-clubs.
Thus, $2ccedvs(G) \leq 2ccedvs(G-y-A)+1$.

Let $\sigma$ be a sequence of length $2ccedvs(G)$ of operations turning $G$ into a disjoint union of 2-clubs.
Because of the geodesics $(a,x,y,z)$, either one of its edges should be deleted, or one of the vertices $\{x,y\}$ should be split.
If the sequence contains a vertex split, we can replace it by an edge deletion on $y-z$.
This new sequence has the same length and is still turning $G$ into a disjoint union of 2-clubs.
Thus, by replacing the vertex split, we get a sequence of operations turning $G-y-A$ into a disjoint union of 2-clubs.
We deduce that $2ccedvs(G) \geq 2ccedvs(G-y-A) +1$ and we conclude that $2ccedvs(G) = 2ccedvs(G-y-A) +1$.
\end{proof}

\begin{definition}
    A 2-club cover of a graph $G$ is a cover $C=\{S_1,S_2,\ldots\}$ of the vertices of $G$, such that for every $S \in C$, $G[S]$ is a 2-club.
    We define the cost of $C$ as follows:
    for every vertex $v$ of $G$, $cost(v,C) = |\{ S \in C | v \in S\}| -1$ and for every edge $u-v$ of $G$, $cost(u-v, C) = 0$ if there exists a set $S \in C$ such that $\{u,v\} \subseteq S$ and 1 otherwise.
    We define the total cost of $C$ as: $cost(C) = \sum_{v \in V} cost(v,C) + \sum_{e \in E} cost(e,C)$. An example is illustrated in Figure \ref{fig:2clubcover}.
\end{definition}

\begin{figure}[h]
    \centering
    \begin{tikzpicture}
    [
        yscale=-1,
        node_style/.style={circle, draw, minimum size=0.5cm} 
    ]
    \def\fscale{5}  \node[node_style] (0) at ($\fscale*(-0.5, 0)$) {$0$};
	 \node[node_style] (1) at ($\fscale*(-0.3, -0.2)$) {$1$};
	 \node[node_style] (2) at ($\fscale*(-0.3, 0.2)$) {$2$};
	 \node[node_style] (3) at ($\fscale*(-0.1, 0)$) {$3$};
	 \node[node_style] (4) at ($\fscale*(0.1, -0.2)$) {$4$};
	 \node[node_style] (5) at ($\fscale*(0.3, -0.2)$) {$5$};
	 \node[node_style] (6) at ($\fscale*(0.1, 0.2)$) {$6$};
	 \node[node_style] (7) at ($\fscale*(0.3, 0.2)$) {$7$};
	 \node[node_style] (8) at ($\fscale*(0.5, 0.2)$) {$8$};

	 \draw (0) to  (1);
	 \draw (1) to  (2);
	 \draw (2) to  (0);
	 \draw (1) to  (3);
	 \draw (3) to  (2);
	 \draw (3) to  (4);
	 \draw (4) to  (5);
	 \draw (3) to  (6);
	 \draw (6) to  (7);
	 \draw (7) to  (8);

     \draw[rounded corners] (-3,-1.5) rectangle (-0.2,1.5);
     \draw[rounded corners] (0.1,0.5) rectangle (3,1.5);
          \draw[rounded corners] (-1,-1.4) rectangle (2,0.4);
\end{tikzpicture} 
    \caption{
    Example of a 2-club cover of a graph with three sets $\{0,1,2,3\}$, $\{3,4,5\}$ and $\{6,7,8\}$. The cost of this cover is 2 because $cost(3) = 1$ and $cost(3,6) = 1$. This cover corresponds to delete edge $3-6$ and to split vertex $3$.    
    }
     \label{fig:2clubcover}
\end{figure}

\begin{lemma}
    The minimum cost of a 2-club cover of a graph $G$ equals $2ccedvs(G)$.
\end{lemma}

\begin{proof}
    Let $C = \{S_1, \ldots, S_c\}$ be a 2-club cover of $G$.
    We delete every edges of cost 1.
    We enumerate the vertices of $G$ as $v_1, \ldots, v_n$.
    For every $i \in \{1, \ldots, n\}$.
    For every $j \in \{1, \ldots, c\}$, if $v_i \in S_j$, we split $v_i$ such that we make a copy $v_{i,j}$ and this copy is adjacent to $N(v_i) \cap S_j$ and $v_i$ is adjacent to $N(v_i) \cap \overline{S_j}$.
    For the last set $S_j$ containing $v_i$ we do not split it because at this point $v_i$ is only adjacent to vertices in $S_j$.
    Therefore, we do $|\{ S \in C | v_i \in S\}| -1$ splits on $v_i$.

    Thus, the sequence is of length $cost(C)$ and it turns $G$ into a disjoint union of 2-clubs because the connected components of the obtained graph correspond to subgraphs of $G$ induced by the sets of $C$.

    Consider a sequence of $G$ turning it into a disjoint union of 2-clubs $G'$.
    We define the cover $C$ as follows.
    For every connected component of $G'$, we define the set of original vertices of each connected component.
    Thus, $C$ is a 2-club cover.
    Each edge deletion corresponds to an edge of cost 1 in the cover.
    The number of times a vertex is split is $|\{ S \in C | v\in S\}|-1$.
    So the length of the sequence is $cost(C)$.
\end{proof}

\begin{lemma} \label{lemma:2ccedvs_rec_multiple}
    Let $y$ be a cut vertex of a graph $G$ such that $A\cup y$ is a 2-club, where $A$ is a connected component in $G\setminus y$.
    By denoting $B = V(G) \setminus (A \cup y)$, we have $2ccedvs(G) = min( 1 + 2ccedvs(G-A), |N(y)\cap B| + 2ccedvs(G-y-A))$. 
\end{lemma}
\begin{proof}
Consider a sequence of modifications turning $G-A$ into a disjoint union of 2-clubs. Split $y$ in $G$ so that one copy is adjacent to $N(y) \cap A$ and one copy is adjacent to $N(y) \cap B$. Since this sequence turns $G$ into a disjoint union of 2-clubs, we have that $2ccedvs(G) \leq 1 + 2ccedvs(G-A)$.

Consider a sequence of modifications which turns $G-y-A$ into a disjoint union of 2-clubs.
Apply this sequence to $G$ and, then delete all edges $y-z$ where $z \in N(y)\cap B$.
Thus, this sequence turns $G$ into a disjoint union of 2-clubs with $|N(y)\cap B|$ additional operations (edge deletions).
Thus, $2ccedvs(G) \leq |N(y)\cap B| + 2ccedvs(G-y-A)$.

Consider a sequence of modifications turning $G$ into a disjoint union of 2-clubs.  
If an edge of $A \cup \{y\}$ is deleted or if a vertex in $A$ is split, we can replace this operation by splitting $y$, with one copy adjacent to $N(y) \cap A$ and the other to $N(y) \cap B$.  
In this case, the restriction of the sequence to the induced subgraph $G - A$ uses one less operation. Therefore, 
$$
2ccedvs(G) \geq 1 + 2ccedvs(G - A).
$$

Let $C$ be a 2-cover of $G$.  
We now prove that there exists a 2-club cover $C'$ of $G$ such that $A \cup \{y\}$ is a set of $C'$, all other sets are in $B \cup \{y\}$, and $\text{cost}(C') \leq \text{cost}(C)$.  
Let $S$ be a set of $C$.  
We claim that $G[S \cap (B \cup \{y\})]$ is a 2-club. Let $i \neq j \in S \cap (B \cup \{y\})$. Then $i$ or $j$ is in $B$.  
Assume without loss of generality that $i \in B$; then $N(i) \subseteq B \cup \{y\}$.  
Since $G[S]$ is a 2-club, there exists $k \in S \cap N[i] \cap N[j]$.  
Hence, $k \in B \cup \{y\}$ because $N(i) \subseteq B \cup \{y\}$.  
This implies $d(i,j) \leq 2$ in $S \cap (B \cup \{y\})$, so $S \cap (B \cup \{y\})$ is a 2-club.

We define $C'$ as follows: for every set $S$ intersecting $B \cup \{y\}$, we add the set $S' = S \cap (B \cup \{y\})$ to $C'$.  
We also add the set $A \cup \{y\}$. As a result, $C'$ is a 2-cover of $G$.

Let us show that $\text{cost}(C') \leq \text{cost}(C)$.  
The cost of vertices in $B$ remains unchanged because they belong to the same sets.  
Similarly, the cost of edges in $B \cup \{y\}$ has not increased, since the sets are restricted accordingly.  
Indeed, let $u-v$ be an edge of $G[B \cup \{y\}]$. If there exists a set $S \in C$ such that $\{u,v\} \subseteq S$, then $\{u,v\} \subseteq S \cap (B \cup \{y\})$.
If all vertices of $A$ and all edges in $A \cup \{y\}$ have cost 0 in $C$, then $A \cup \{y\}$ must already be a set in $C$: it is the set containing $A$.  
Hence, $\text{cost}(C') \leq \text{cost}(C)$ because the vertices of $A$ and edges of $A \cup \{y\}$ still have cost 0, and $y$ appears in fewer sets.

Otherwise, some vertex of $A$ or some edge in $A \cup \{y\}$ has cost at least 1 in $C$.  
This element now has cost 0 in $C'$. Furthermore, the cost of $y$ increases by at most 1 due to the addition of the set $A \cup \{y\}$.  
Consequently, $\text{cost}(C') \leq \text{cost}(C)$.

We conclude that there exists a minimum 2-club cover of $G$ such that $A \cup \{y\}$ is a set and all other sets are in $B \cup \{y\}$.  
Consider such a minimum 2-club cover $C$ of $G$. If $\text{cost}(y,C) = 0$, then $y$ appears in only one set of $C$.  
In that case, all edges $y-b$ with $b \in B$ are of cost 1, and the remaining sets lie in $B$.  
Therefore,
$$
2ccedvs(G) \geq |N(y) \cap B| + 2ccedvs(G - y - A).
$$
Otherwise, if $\text{cost}(y,C) \geq 1$, then 
$$
2ccedvs(G) \geq 1 + 2ccedvs(G - A).
$$

\end{proof}

\begin{theorem}
\label{theorem:2CCEDVS_trees}
    {\sc 2CCEDVS} is solvable on polynomial time in trees. 
\end{theorem}

\begin{proof}
Consider a tree $T$ with $n$ vertices and a postorder numbering from $1$ to $n$ of the vertices such that the deepest branches are visited first.
    
For every vertex $x$, we denote by $\phi(x)$ the minimum descendant of $x$, that is, the descendant with the smallest postorder numbering.
For any $i$ and $j$ such that $j$ is an ascendant of $i$, we define $T[i,j]$ as the induced subgraph of $T$ from the vertices $\{i, i+1, \ldots, j\}$.
We define $t[i,j] = 2ccedvs(T[i,j])$.
Thus, $2ccedvs(T) = t[1,n]$.
See Figure~\ref{fig:tree_postorder} for an example.

    \begin{figure}[!htb] 
        \centering
    \begin{tikzpicture}
    [
        yscale=-1,
        node_style/.style={circle, draw, minimum size=0.5cm} 
    ]
    \def\fscale{5} 
    \node[node_style] (1) at ($\fscale*(-0.28, 0.28)$) {$1$};
	 \node[node_style] (2) at ($\fscale*(-0.39, 0.06)$) {$2$};
	 \node[node_style] (3) at ($\fscale*(-0.17, 0.06)$) {$3$};
	 \node[node_style] (4) at ($\fscale*(0.06, 0.06)$) {$4$};
	 \node[node_style] (5) at ($\fscale*(-0.17, -0.17)$) {$5$};
	 \node[node_style] (6) at ($\fscale*(0.28, 0.28)$) {$6$};
	 \node[node_style] (7) at ($\fscale*(0.5, 0.28)$) {$7$};
	 \node[node_style] (8) at ($\fscale*(0.39, 0.06)$) {$8$};
	 \node[node_style] (9) at ($\fscale*(0.28, -0.17)$) {$9$};
	 \node[node_style] (10) at ($\fscale*(0.06, -0.28)$) {$10$};
	 \node[node_style] (0) at ($\fscale*(-0.5, 0.28)$) {$0$};

	 \draw (0) to  (2);
	 \draw (2) to  (5);
	 \draw (5) to  (10);
	 \draw (10) to  (9);
	 \draw (9) to  (8);
	 \draw (8) to  (6);
	 \draw (8) to  (7);
	 \draw (4) to  (5);
	 \draw (3) to  (5);
	 \draw (1) to  (2);
\end{tikzpicture}

\caption{Example of a tree and its postorder numbering ordered by decreasing depth. For example the children of $8$ are $6$ and $7$ and $\phi(9) = 6$ (the minimum descendant of $9$).}
        \label{fig:tree_postorder}
    \end{figure}

Consider two integers $i,j$ in $\{1,1,\ldots, n\}$ where $i<j$ and, $j$ is an ascendant of $i$.
We denote by $x$ the parent of $i$ and by $y$ the parent of $j$.
As $x$ is the parent of $i$, $y$ is a cut vertex of $T[i,j]$ separating the vertices $A = \{i, \ldots, x\}$ from the rest of the tree.
As the numbering is in decreasing depth order, $T[i,x]$ is a star centered at $x$ (all the children of $x$ are leaves) and, $T[A]$ is a 2-club.
    
If $degree(y) = 2$ (See Figure~\ref{fig:2ccedvs_trees_algo_simple}) and $z$ is the other child of $y$, then by Lemma~\ref{lemma:2ccedvs_rec_simple}, we have
    \[ t[i,j] =
    \begin{cases}
        1 + t[y+1, j], & \text{if } y \not=j \\
        1 + t[x+1, z], & \text{otherwise}.
    \end{cases}
    \]

  \begin{figure}[!htb] 
        \centering
        \begin{subfigure}[b]{0.49\textwidth}
            \centering
\begin{tikzpicture}
    [
        yscale=-1,
        node_style/.style={circle, draw, minimum size=0.5cm} 
    ]
    \def\fscale{2} 
    
	 \node[node_style] (1) at ($\fscale*(0, 0.5)$) {};
	 \node[node_style] (2) at ($\fscale*(-0.25, 0)$) {$x$};
	 \node[node_style] (0) at ($\fscale*(-0.5, 0.5)$) {$i$};
	 \node[node_style] (3) at ($\fscale*(0, -0.5)$) {$y$};
	 \node[node_style] (4) at ($\fscale*(0.25, 0)$) {$z$};
	 \node[node_style, label={below:$x+1$}] (5) at ($\fscale*(0.5, 0.5)$) {};

	 \draw (0) to  (2);
	 \draw (1) to  (2);
	 \draw (2) to  (3);
	 \draw (3) to  (4);
	 \draw (4) to  (5);
\end{tikzpicture}

\caption{Case where $y=j$, so $y$ has no ascendant in $T[i,j]$ and $z$ is the other child of $y$.}
            \label{fig:2ccedvs_trees_algo_simple1}
        \end{subfigure}
        \hfill
        \begin{subfigure}[b]{0.49\textwidth}
            \centering
\begin{tikzpicture}
    [
        yscale=-1,
        node_style/.style={circle, draw, minimum size=0.5cm} 
    ]
    \def\fscale{3} 
    
	 \node[node_style] (1) at ($\fscale*(-0.17, 0.5)$) {};
	 \node[node_style, ] (2) at ($\fscale*(-0.33, 0.17)$) {$x$};
	 \node[node_style] (0) at ($\fscale*(-0.5, 0.5)$) {$i$};
	 \node[node_style] (3) at ($\fscale*(-0.17, -0.17)$) {$y$};
	 \node[node_style] (4) at ($\fscale*(0, -0.5)$) {$z$};
	 \node[node_style] (5) at ($\fscale*(0.17, -0.17)$) {};
	 \node[node_style] (6) at ($\fscale*(0.5, 0.17)$) {};
	 \node[node_style] (7) at ($\fscale*(0.17, 0.17)$) {};
	 \node[node_style, label={below:$y+1$}] (8) at ($\fscale*(0.17, 0.5)$) {};
	 \node[node_style] (9) at ($\fscale*(0.4, 0.5)$) {};

	 \draw (0) to  (2);
	 \draw (1) to  (2);
	 \draw (2) to  (3);
	 \draw (3) to  (4);
	 \draw (4) to  (5);
	 \draw (5) to  (6);
	 \draw (5) to  (7);
	 \draw (7) to  (8);
	 \draw (7) to  (9);
\end{tikzpicture}

\caption{Case where $y <j$, so $y$ has only $x$ as a child and $z$ is the parent of $y$.}
            \label{fig:2ccedvs_trees_algo_simple2}
        \end{subfigure}
        \caption{Example where $degree(y) = 2$. In this case, the edge $y-z$ can be deleted.}
        \label{fig:2ccedvs_trees_algo_simple}
    \end{figure}

If $degree(y) \geq 3$, then $2ccedvs(T[i,j]-A) = t[x+1, j]$, and $T-y-A$ is the disjoint union of $T[y+1,j]$ and $T[\phi(k), k]$ where $k$ is a child of $y$ (see Figure~\ref{figure:2ccedvs_trees_algo_multiple}). According to Lemma~\ref{lemma:2ccedvs_rec_multiple}, $2ccedvs(T[i,j]) = \min(1 + 2ccedvs(T-A), degree(y)-1 + 2ccedvs(T-y-A))$.
    Thus, 
    \[
    t[i,j] = \min(1 + t[x+1,j], degree(y)-1 + t[y+1,j] + \sum_{k \in children(y), k\not=x} t[\phi(k),k])
    \]
 
We obtain an $O(n^2)$ algorithm via a recursive approach.

\begin{figure}[!htb] 
\centering

\begin{tikzpicture}
    [
        yscale=-1,
        node_style/.style={circle, draw, minimum size=0.5cm} 
    ]
    \def\fscale{5} 
	 \node[node_style] (1) at ($\fscale*(-0.24, 0.39)$) {};
	 \node[node_style,] (2) at ($\fscale*(-0.37, 0.13)$) {$x$};
	 \node[node_style] (0) at ($\fscale*(-0.5, 0.39)$) {$i$};
	 \node[node_style] (3) at ($\fscale*(-0.24, -0.13)$) {$y$};
	 \node[node_style] (4) at ($\fscale*(0.02, -0.39)$) {};
	 \node[node_style] (5) at ($\fscale*(0.28, -0.13)$) {};
	 \node[node_style, label={below:$y+1$}] (7) at ($\fscale*(0.5, 0.13)$) {};
	 \node[node_style] (10) at ($\fscale*(-0.11, 0.13)$) {$k_1$};
	 \node[node_style] (8) at ($\fscale*(0.15, 0.13)$) {$k_2$};
	 \node[node_style, label={below:$\phi(k_1)$}] (9) at ($\fscale*(-0.11, 0.39)$) {};
	 \node[node_style] (12) at ($\fscale*(0.02, 0.39)$) {};
	 \node[node_style, label={below:$\phi(k_2)$}] (13) at ($\fscale*(0.15, 0.39)$) {};
	 \node[node_style] (14) at ($\fscale*(0.28, 0.39)$) {};

	 \draw (0) to  (2);
	 \draw (1) to  (2);
	 \draw (2) to  (3);
	 \draw (3) to  (4);
	 \draw (4) to  (5);
	 \draw (5) to  (7);
	 \draw (3) to  (10);
	 \draw (3) to  (8);
	 \draw (10) to  (9);
	 \draw (10) to  (12);
	 \draw (8) to  (13);
	 \draw (8) to  (14);
\end{tikzpicture}

\caption{Example where $degree(y) \geq 3$. In this case, either we delete the edges incident to $y$ except $x-y$ and or we split $y$ to make $\{i, \ldots, x, y\}$ a cluster. Both of these operations would be followed by a recursive call.}
        \label{figure:2ccedvs_trees_algo_multiple}
    \end{figure}

\end{proof}

\subsection{Relationship between 2ccvs and 2ccedvs.}

Observe that deleting an edge $u-v$ can be replaced by splitting both $u$ and $v$ so that the copy of $u$ has only $v$ in its neighborhood and vice-versa. Therefore:

\begin{theorem} For every graph $G$, $2ccedvs(G) \leq 2ccvs(G) \leq 2 \cdot 2ccedvs(G)$. 
\end{theorem}

Furthermore the above inequality is tight because $2ccedvs(P_5) = 1$ and $2ccvs(P_5) = 2$.
Moreover, although of interest by itself, the inequality implies that \TCCEDVS{} admits a factor-two approximation algorithm on classes of graphs where \TCCVS{} is solvable in polynomial-time.

\section{Hardness of Approximation}

Our objective in this section is to adapt the construction in Section \ref{sec:2CCVS} to reduce \textsc{MAX 3-SAT(4)} to {\sc 2CCVS}. \textsc{MAX 3-SAT(4)} is a variant of MAX 3-SAT where each variable appears at most four times in the given formula $\phi$.

We also add the following constraint: when a variable $V_i$ appears exactly two times positively and two times negatively, we suppose that the list of clauses in which $V_i$ appears, $C(V_i)_1, C(V_i)_2,C(V_i)_3,C(V_i)_4$, is ordered so that $V_i$ appears positively in $C(V_i)_1$ and $C(V_i)_3$ and negatively in $C(V_i)_2$ and $C(V_i)_4$.  
This constraint is added to ensure that each unsatisfied clause in $\phi$ causes an additional split in the construction.
In fact, we can observe that if the formula $\phi$ cannot be satisfied, then we can use an ``extra'' split in each clause gadget to obtain a solution. 
However, the inverse does not necessarily hold if there is a variable $V_i$ that occurs two times positively and two times negatively.

Indeed, by using $12+1$ splits in the variable cycle, we may be able to satisfy the four clauses where $v$ occurs. In the rest of this section, we show how to obtain a reduction from {\sc MAX 3-SAT(4)} to {\sc 2CCVS} to prove the below theorem.

\begin{theorem}\label{th:inapprox}
{\sc 2CCVS} is $\APX$-hard.
\end{theorem}

We will make use of linear reductions defined as follows:
\begin{definition}[Linear reduction \cite{PapadimitriouY91}]
    Let $A$ and $B$ be two optimization problems with cost functions $cost$ and $cost'$.
    We denote by $OPT(I)$ the optimal value of an optimization problem on an instance $I$.
    We say that $A$ has a linear reduction to $B$ if there exist two polynomial time algorithms $f$ and $g$ and two positive numbers $\alpha$ and $\beta$ such that for any instance $I$ of $A$:
    \begin{itemize}
        \item $f(I)$ is an instance of $B$.
        \item $OPT(f(I)) \leq \alpha OPT(I)$ .
        \item For any solution $S'$ of $f(I)$, $g(S')$ is a solution of $I$ and   $|cost(g(S'))- OPT(I)| \leq \beta |cost(S') - OPT(f(I))|$.
    \end{itemize}
\end{definition}

\begin{proof}
  
  Recall that in an optimal solution of {\sc MAX 3-SAT(4)}, at least $\frac{7}{8}$ of the clauses are satisfied~\cite{Hastad97}, yielding
  \begin{equation}
    OPT(\phi) \geq \frac{7M}{8}.\label{eq:7m by 8}
  \end{equation}

Let $f$ be the function transforming any instance $\phi$ of {\sc MAX 3-SAT(4)} into a graph $G$ with Construction~\ref{construction:2ccvs}.

We can find a sequence of splits in $G$ turning it into a disjoint union of 2-clubs with $3a(V)$ splits (by splitting one over two vertices of each variable cycle) for every variable cycle and $3$ splits for every clause (by splitting the three vertices of the clause triangle) Thus, 
\begin{align*}
OPT(G) &\leq 3M + 9M = 12 M \\
\text{As } M &\leq \frac{8}{7} OPT(\phi) \text{, we deduce that} \\
OPT(G) &\leq 12 \frac{8}{7} OPT(\phi) =\frac{96}{7} OPT(\phi)
\end{align*}

Consider an assignment of $\phi$ maximizing the number of satisfied clauses of $\phi$.
Let us define a sequence of operations on $G_\phi$.
For every true (resp., false) variable $v$, we split the vertices $v_{2i}$ (resp., $v_{2i+1}$) for every $i$.
Let $C$ be a clause.
If $C$ is satisfied, consider $v$ a variable satisfying $C$.
Split the vertices $u_C$ and $w_C$ (like in the proof of Theorem~\ref{theorem:2CCEDVS_NPC}) where $u$ and $w$ are the two other variables in $C$.
Otherwise, we split the three vertices $u_C$, $v_C$ and $w_C$.
We denote by $SC$ the set of satisfied clauses of $\phi$ by $\sigma$ and $UC$ the set of unsatisfied clauses of $\phi$ by $\sigma$.
This sequence is of length $\sum_{v \in V} 3 a(v) + \sum_{c \in SC} 2 + \sum_{c \in UC} 3 = 9M + 2M + k$ where $k$ is the number of unsatisfied clauses.
Furthermore, this sequence of operations turns the graph $G_\phi$ into a disjoint union of bicliques.
We deduce that $OPT(G_\phi) \leq 12M + M-OPT(\phi)$.
Thus, $OPT(\phi) + OPT(G_\phi) \leq 12M + M$.

Let $X$ be a sequence of operations turning the graph into a disjoint union $2$-clubs.
Let us define an assignment $g(X)$.
Consider a variable $v$.
According to the proof of Theorem~\ref{theorem:2CCEDVS_NPC}, the variable cycle of $v$ needs at least $3a(v) = 12$ operations on its edges and its vertices.

If the variable cycle is using exactly $12$ operations, then there are two cases.
If the vertices $v_{2i}$ are split for every $i$, then we set $v$ to true and otherwise to false.

If the variable cycle is using at least $12+2$ operations, then we replace these operations by splitting the vertices $v_{2i+1}$ for every $i$ and by splitting the vertices $v_1$ and $v_{13}$ so that it disconnects the variable cycle and the two positive clauses connected to $v$.
We set $v$ to false.

If the variable cycle is using at $12+1$ splits, then we show that it is not possible that it resolves all geodesics of the cycle with the clause edges incident to $v_1$, $v_8$, $v_{13}$ and $v_{19}$.
If $v$ is resolving $C_{v_1}$ and $C_{v_3}$, then we set $v$ to true.
Otherwise, we set $v$ to false.

As $cost(g(X))$ is the number of satisfied clauses of $g(X)$, we have $13M \leq cost(X) + cost(G(X))$ because for each unsatisfied clause, we can add only one split to satisfy it.
Therefore we have:

\begin{align*}
OPT(\phi) + OPT(G_\phi) &\leq 13M  \leq cost(X) + cost(g(X)) \\
OPT(\phi) - cost(g(X)) &\leq cost(X) - OPT(G_\phi)
\end{align*}

Thus, $|cost(g(X)) - OPT(\phi)| \leq |cost(X) - OPT(G)|$ because {\sc MAX-3SAT(4)} is a maximization problem while {\sc 2CCVS} is a minimization problem.
We have constructed a linear reduction with $\alpha=\frac{96}{7}, \beta=1$.
As MAX 3-SAT(4) is \APX-hard, 2CCVS is also \APX-hard.
\end{proof}

In the same way, we can prove:

\begin{theorem}
{\sc 2CCEDVS} is $\APX$-hard.
\end{theorem}

\section{The Parameterized Complexity of \TCCEDVS{} and \TCCVS{}}

As observed for {\sc Cluster Editing with Vertex Splitting} in \cite{abukhzam2023cluster}, all edge deletions can be performed before vertex splits.
Thus, we assume that any sequence of operations is equivalent to a sequence of operations where the splitting is performed at the end. 
Our proof is based on branching on paths of length three whose endpoints are at distance exactly three from each other.
 
In the case/branch where a vertex 
is to be split, we simply mark it for splitting and perform this operation at the end, until no such length-three paths exist. This ``charge and reduce'' approach is explained in more detail in the sequel.

\begin{lemma}\label{lemma:allornothing}
Consider a minimal sequence of edge deletions and vertex splits.
Let $S$ be the set of the split vertices. If $v \in S$ and $C$ is a connected component of $G[ V(G) \setminus S]$, then each copy of $v$ is either adjacent to all the vertices of $C\cap N(v)$ or to none of them.

\end{lemma}

\begin{proof}
Suppose there exists a copy $u$ of $v$ and $x,y$ two neighbors of $v$ in $C$ such that $u$ is adjacent to $x$ and not to $y$. As $y$ is adjacent to $v$, then there exists another copy $u'$ of $v$ such that $u'$ is adjacent to $y$. As $x$ and $y$ are connected in $G[V(G) \setminus S]$, then $u$ and $u'$ belong to the same connected component of the resulting solution, which is impossible since any two copies of a split vertex must belong to two different clubs, unless the solution is not minimal, a contradiction.

\end{proof}

\begin{lemma}
Let $G=(V,E)$ be a connected graph and assume a sequence of at most $k$ edge deletions and splits is applied to $G$. 
If $S$ is the set of split vertices, then $G[V \setminus S]$ has at most $k+1$ connected components.
\end{lemma}

\begin{proof}
The statement simply follows from the fact that the graph is initially connected and a single edge deletion operation or a single split operation can only increase its number of connected components by at most 1.

\end{proof}

We now show how to perform the vertex splits at the end, after marking the vertices to be split and exhausting all the possible edge deletion operations. Given a split set $S = \{v_1, \ldots, v_s\}$, $p$ connected components $C=\{C_1, \ldots, C_p\}$, and a number of possible extra splits $e$.
We consider the algorithm $Aux(S, C, e)$ for trying all sequences of splits on $S$ of length at most $s+e$ such that each vertex of $S$ is split at least once.

At the end of such a sequence, the split set is of size at most $2s+e$ (since we create at most $s+e$ copies).
A split is a choice of a vertex in the current split set and one subset of the vertices for each of the two copies.
There are at most $2s+e$ choices to select a vertex that will be split.
Choosing a set of neighbors for both copies of the split corresponds to selecting two subsets of the current split set and of the connected components (see Lemma \ref{lemma:allornothing}).
There are at most $2^{2s+e+p}$ choices for the set of neighbors of one copy.
Therefore, there are at most $((2s+e) 2^{4s+2e+2p} )^{s+e}$ sequences of such splits.

\begin{theorem}
    The complexity of Algorithm \textsc{Aux} is 
    in $O((3k 2^{8k})^{2k})$ if $p \leq k+1$, $e+s \leq k$.
\end{theorem}

We now describe the main algorithm, which returns true if there exists a sequence of length at most $k$ of edge deletions or vertex splits on a graph which results in a disjoint union of $2$-clubs and returns false otherwise.
In the following, $G=(V,E)$ is a graph, $S\subseteq V$ is the set of vertices that were marked for splitting, and $k$ is the number of allowed splits.

\begin{algorithm}[!htb]
\caption{$f(G,S,k)$}
\label{2CCEDVSalg}
\begin{algorithmic}[1]
\If{$k=0$}
    \State \Return $Aux(S, G[V(G) \setminus S], 0)$
\ElsIf{there exists $x,y$ such that $d(x,y) = 3$ and $u$ and $v$ are interior vertices of a length 3 path between $x$ and $y$ that are not in $S$} 
    \If{$f(G - xu,S, k-1)$}
         \Return true
    \EndIf
    \If{$f(G - uv,S, k-1)$}
         \Return true
    \EndIf
    \If{$f(G -vy,S, k-1)$}
         \Return true
    \EndIf
    \If{$f(G,S \cup u, k-1)$}  \Return true \EndIf
    \If{$f(G,S \cup v, k-1)$}
         \Return true
    \EndIf
\Else 
    \State \Return $Aux(S, G[V(G)\setminus S], k)$

\EndIf

\end{algorithmic}
\end{algorithm}

The algorithm branches to 5 sub-problems and each step is taking $O(n^3)$ to search for a pair of vertices at distance $3$ from each other. Therefore, we deduce:

\begin{theorem}
Algorithm \ref{2CCEDVSalg} 
solves {\sc 2CCEDVS} in $O(n^3 5^k (3k 2^{8k})^{2k})$. Hence, the problem
is fixed-parameter tractable with respect to the total number of edge deletions and vertex splits.
\end{theorem}

Note that a geodesic of length three is also an obstruction in the case of the {\sc 2CCVS} problem. Therefore, we can adapt the previous algorithm by removing the branchings where we delete edges (just remove lines 4-9 of Algorithm \ref{2CCEDVSalg}). We conclude that:

\begin{corollary}
{\sc 2CCVS} is
fixed-parameter tractable 
with respect to the solution size.
\end{corollary}

\section{Concluding Remarks}

This paper introduces the {\sc 2-club Cluster Vertex Splitting} (2CCVS) and {\sc 2-club Cluster Edge Deletion with Vertex Splitting} (2CCEDVS) problems. Both problems are shown to be $\NP$-complete in general and that they remain $\NP$-complete on planar graphs of maximum degree four and three (respectively).
We further considered the polynomial-time approximability of the two problems and showed them to be $\APX$-hard. We believe a constant-factor approximation for {\sc 2CCVS} is not too difficult to obtain. In fact, this remains an interesting open problem in both cases.

On the positive side, we showed that both {\sc 2CCVS} and {\sc 2CCEDVS} are fixed-parameter tractable when parameterized by the number of allowed modifications. We believe that obtaining an improved algorithm that runs in $O^*(c^k)$ for {\sc 2CCVS} is not too difficult since a simple branching algorithm would have two cases for each length-three ``obstruction path'' and the rest consists of performing vertex splits only. However, in the case of {\sc 2CCEDVS}, obtaining an algorithm with a running time in $O^*(c^k)$ seems much more challenging, and we pose it here as an open problem.
Furthermore, we gave polynomial-time algorithms for both {\sc 2CCVS} and {\sc 2CCEDVS} when the input is restricted to trees. 
This suggests considering the parameterized complexity of the two problems when parameterized by treewidth.

Whether the obtained hardness results for 2-clubs hold also for $s$-clubs (i.e. for any/all $s\geq 3$) remains an open question, to be explored in future work. Other directions that could be explored include the parameterized complexity of the problems with respect to other parameters such as the treewidth of the graph, as well as using additional local parameters such as the number of times a vertex can split, which seems to be a realistic constraint. In fact, the use of local parameters seems to be an effective approach in the context of cluster editing \cite{barr2019combinatorial,komusiewicz2012cluster}.
Finally, an interesting open problem is whether the two considered problems admit polynomial-size kernels.

\end{document}